\documentclass[3p,preprint,sort&compress]{elsarticle}
\usepackage{eurosym}
\usepackage{amsmath}
\usepackage{amssymb}
\usepackage{amsfonts}
\usepackage{epsfig}
\usepackage{subcaption}
\usepackage{graphicx}
\usepackage{multirow}
\usepackage{lipsum}
\usepackage{caption}
\usepackage[usenames,dvipsnames]{xcolor}
\usepackage{tikz}
\usepackage{float}

\newtheorem{theorem}{Theorem}

\newtheorem{lemma}{Lemma}

\newenvironment{proof}[1][Proof]{\noindent\textbf{#1.} }{\ \rule{0.5em}{0.5em}}

\journal{ArXiv}
\begin{document}

\begin{frontmatter}
	
\title{Discrimination in Heterogeneous Games}
\author[UPV1,IKE]{Annick Laruelle\corref{cor1}}
\ead{annick.laruelle@ehu.eus}
\author[PUC]{Andr\'e Rocha}

\address[UPV1]{Department of Economic Analysis (ANEKO), University of the Basque Country
	(UPV/EHU);\\ Avenida Lehendakari Aguirre, 83, E-48015 Bilbao, Spain.}
\address[IKE]{IKERBASQUE, Basque Foundation of Science, 48011, Bilbao, Spain.}
\address[PUC]{Department of Industrial Engineering,	Pontifical Catholic University of Rio de Janeiro,\\
	Rua Marqu\^es de S\~ao Vicente 225, G\'avea, CEP22451-900, Rio de Janeiro, RJ, Brazil.}

\cortext[cor1]{Corresponding author}
	
	\begin{abstract}
		In this paper, we consider coordination and anti-coordination heterogeneous games played by
		a finite population formed by different types of individuals who fail to
		recognize their own type but do observe the type of their opponent. We show
		that there exists symmetric Nash equilibria in which players discriminate by acting 
		differently according to the type of opponent that they face in anti-coordination games, while no
		such equilibrium exists in coordination games. Moreover, discrimination 
		has a limit: the maximum number of groups where the treatment differs is three. We then 
		discuss the theoretical results in light of the observed behavior of people in some specific 
		psychological contexts. 
	\end{abstract}


\end{frontmatter}

\section{Introduction}

Heterogeneous games have been studied in I\~{n}arra and Laruelle (2012),
Barreira da Silva Rocha and Laruelle (2013) and I\~{n}arra, Laruelle and
Zuazo-Garin (2015). These are games where the population is divided into
types. A type is some particular feature that makes the individual to be
perceived as different by the other individuals in the population. Examples
could be a genotype, phenotype or behavior. The division is artificial in
the sense that individuals have the same capacities and play the same game.
But types can be distinguished and players can adapt their behavior
according to the type of the opponent. The second characteristic of such
heterogeneous games is that individuals lack self-perception (that is, they
do not know their type) but recognize the type of the opponent.

In this paper, we study how individuals discriminate different types of
opponents in symmetric Nash equilibria of heterogeneous games. We show that
in coordination games players do not behave differently when they face
different types of opponents. By contrast in anti-coordination games,
equilibria in which players behave differently according to the type of
opponent arise. In other words, an artificial division of the population may
generate a real discrimination.

In anti-coordination games with three types, there are three kinds of
equilibria: non discriminating equilibria, partially discriminating
equilibria (where two types are treated equally and one is treated
differently) and totally discriminating equilibria (where each of the three
types is treated differently). The following question is whether this can be
generalized when there are more than three types. The answer is negative:
the maximum number of groups where the treatment differs is three. That is,
discrimination has a limit: there is no totally discriminating equilibrium
when the population is divided into four or more types.

The rest of the paper is organized as follows: in Section 2 we present the
model; in Section 3, we obtain all possible symmetric Nash equilibria and
Section 4 concludes with a discussion.

\section{Heterogeneous games with m types}

Consider a population of $n$ individuals of $m\geq 2$ different types. Let $%
T $ be the set of types. The proportion of individuals of type $t$ in $T$ is
given by $x_{t}$ ($\sum_{t\in T}x_{t}=1$) with $1<nx_{t}<n$. That is, we
assume that there is strictly more than one individual of each type.

Any pair of individuals (whatever their types) plays the same symmetric game
with the normalized matrix\footnote{%
Details of this normalization are given in Eichberger, Haller and Milne
(1993).}: 
\begin{equation}
\begin{tabular}{c|cc}
& Cooperate & Defect \\ \hline
Cooperate & $0$ & $y$ \\ 
Defect & $z$ & $0$%
\end{tabular}
\label{Matrix}
\end{equation}
This corresponds to a matrix of coordination (it is better if players choose
the same action) if $y,z<0$ and a matrix of anti-coordination (it is better
if players choose different actions) if $y,z>0$. Let 
\begin{equation*}
\zeta =\frac{y}{y+z}.
\end{equation*}

The following assumptions are made concerning the types: (i) Players do not
know their own type; and (ii) Players recognize their opponents' type. In
this context, individuals cannot condition their behavior on their type. By
contrast they can choose different probabilities of playing ``cooperate''
when they face different types of opponents. A \textit{strategy} can thus be
represented by $(\alpha _{t})_{t\in T}$, where $\alpha _{t}$ gives the
probability of cooperation when facing an individual of type $t$.

A strategy $(\alpha _{t})_{t\in T}$ is said to be $l$-discriminating if
there exists an $l$-partition\footnote{$T=T_{1}\cup ..\cup T_{l}$ and $%
T_{j}\cap T_{k}=\emptyset $ for any $j\neq k$.} of $T$, $(T_{1},...T_{l})$
such that $\alpha _{s}=\alpha _{t}$ for any $s,t\in T_{k}$, and $\alpha
_{s}\neq \alpha _{t}$ if $s\in T_{j}$ and $t\in T_{k}$. In this case the
probability of cooperation when facing an individual of type $s\in T_{k}$ is
denoted $\alpha _{k}$. The special case $l=1$ corresponds to $\alpha
_{i}=\alpha $ for any $t\in T$ and the strategy is referred to as\textit{\
non-discriminating strategy as} the probability of cooperation is identical
whatever the type of the opponent. When $l=m$ the strategy is said to be 
\textit{totally discriminating} as the probability of cooperation is
different for each type of opponent: $\alpha _{i}\neq \alpha _{j}$ for any $%
i,j\in T$.

A pair of individuals is selected at random to play the game. This is
equivalent to assuming that an individual is randomly picked from a set of
size $n$, and then an opponent is picked from the remaining set of size $n-1$
. First consider that the individual is of type $i\in T$ and an opponent of
type $j\in T$. The probability of the encounter is denoted $p(i,j)$. We have 
\begin{equation}
p(i,j)=\frac{nx_{i}x_{j}}{n-1}\text{ for }i\neq j\text{ and }p(i,i)=\frac{
(nx_{i}-1)x_{i}}{n-1}\text{.}  \label{Probabilities}
\end{equation}
If an individual plays $(\alpha _{t})_{t\in T}$ while the opponent plays $%
(\beta _{t})_{t\in T}$, the individual sees that the opponent is of type $j$
(and plays $\alpha _{j}$) while the opponent sees that the individual is of
type $i$ (and plays $\beta _{i}$). The payoff obtained by the individual in
this encounter is obtained from matrix (\ref{Matrix}). It is given by 
\begin{equation*}
y\alpha _{j}(1-\beta _{i})+z(1-\alpha _{j})\beta _{i}=z\beta _{i}+\left[
y-(y+z)\beta _{i}\right] \alpha _{j}.
\end{equation*}
This payoff is weighted by the probability of occurrence of the encounter, $%
p(i,j)$; and the payoff of all possible encounters are summed to obtain the
total expected payoff of the individual.

The \textit{expected payoff} of an individual who plays $(\alpha _{t})_{t\in
T}$ while the opponent plays $(\beta _{t})_{t\in T}$ is denoted by $%
U((\alpha _{t},\beta _{t})_{t\in T})$. We have 
\begin{equation}
U((\alpha _{t},\beta _{t})_{t\in T})=\sum_{j\in T}\sum_{i\in T}p(i,j)\left[ z%
\text{ }\beta _{i}+y\text{ }\alpha _{j}-(y+z)\beta _{i}\alpha _{j}\right] .
\label{Ut}
\end{equation}%
Substituting (\ref{Probabilities}) in (\ref{Ut}), we obtain after some
algebra (see the Appendix for details): 
\begin{equation}
U((\alpha _{t},\beta _{t})_{t\in T})=z\sum_{j\in T}x_{j}\text{ }\beta
_{j}+\sum_{j\in T}\frac{x_{j}}{n-1}\left[ (n-1)y+(y+z)\beta
_{j}-n(y+z)\sum_{i\in T}x_{i}\beta _{i}\right] \alpha _{j}.  \label{Ut2}
\end{equation}%
The heterogeneous game is denoted $\Gamma (n,T,(x_{t})_{t\in T},y,z)$.

\section{Symmetric Nash equilibria with m-types}

Here we look for the symmetric Nash equilibria, that is equilibrium with $%
\alpha _{t}=\beta _{t}$ for any $t\in T$. We denote such an equilibrium by $%
(\alpha _{t}^{\ast })_{t\in T}$ (instead of $((\alpha _{t}^{\ast })_{t\in
T},(\alpha _{t}^{\ast })_{t\in T})$). Whenever the strategies at equilibrium
are $l$-discriminating, the equilibrium will be referred to as $l$%
-discriminating equilibrium. In particular we will refer to non
discriminating equilibrium and totally discriminating equilibrium.

To find the equilibria in game $\Gamma (n,T,(x_{t})_{t\in T},y,z)$ we
re-write (\ref{Ut2}) as 
\begin{equation*}
U((\alpha _{t},\beta _{t})_{t\in T})=z\sum_{j\in T}x_{j}\text{ }\beta
_{j}+\sum_{j\in T}\frac{x_{j}}{n-1}\mathcal{F}_{j}((\beta _{t})_{t\in
T})\alpha _{j}
\end{equation*}%
where 
\begin{equation}
\mathcal{F}_{j}((\beta _{t})_{t\in T})=(n-1)y+(y+z)\beta
_{j}-n(y+z)\sum_{i\in T}x_{i}\beta _{i}.  \label{FI(beta)}
\end{equation}

The individual chooses $(\alpha _{k})_{k\in T}$ in response to the opponent
playing $(\beta _{t})_{t\in T}$. When the individual sees that the opponent
is of type $k$ the best choice for $\alpha _{k}$ depends on the sign of $%
\mathcal{F}_{k}((\beta _{t})_{t\in T})$. An individual should choose $\alpha
_{k}=1$ when $\mathcal{F}_{k}((\beta _{t})_{t\in T})>0$ and $\alpha _{k}=0$
when $\mathcal{F}_{k}((\beta _{t})_{t\in T})<0$. If $\mathcal{F}_{k}((\beta
_{t})_{t\in T})=0$ any $\alpha _{k}$ can be chosen.

At equilibrium $(\alpha _{t}^{\ast })_{t\in T}$ the following conditions
hold: 
\begin{equation}
\begin{array}{cl}
\alpha _{k}^{\ast }=1 & \text{when }\mathcal{F}_{k}((\alpha _{t}^{\ast
})_{t\in T})>0 \\ 
\alpha _{k}^{\ast }=0 & \text{when }\mathcal{F}_{k}((\alpha _{t}^{\ast
})_{t\in T})<0 \\ 
0<\alpha _{k}^{\ast }<1 & \text{when }\mathcal{F}_{k}((\alpha _{t}^{\ast
})_{t\in T})=0.%
\end{array}
\label{EquilibriumConditions}
\end{equation}%
The first theorem gives the non discriminating equilibria. That is,
equilibria $(\alpha _{t}^{\ast })_{t\in T}$ with $\alpha _{t}^{\ast }=\alpha
^{\ast }$ for any $t\in T$. The proofs of all theorems are given in the
Appendix.

\begin{theorem}
\label{TheoremNonDiscriminatingequilibrium}Consider a heterogeneous game $%
\Gamma (n,T,(x_{t})_{t\in T},y,z)$. $(\alpha ^{\ast },...,\alpha ^{\ast })$
is a symmetric equilibrium if and only if one of the following condition is
satisfied:

\begin{enumerate}
\item $y,z<0$, and $[\alpha ^{\ast }=\zeta $; $\alpha ^{\ast }=1$; or $%
\alpha ^{\ast }=0]$;

\item $y,z>0$, and $\alpha ^{\ast }=\zeta $.
\end{enumerate}
\end{theorem}

That is, if individuals do not take into account the different types, we
obtain the classical results. In coordination games there are three
equilibria, two in pure strategies and one in mixed strategies. In
anti-coordination games there is only one equilibrium, in mixed strategies.%
\footnote{%
The other two equilibria in pure strategies are not symmetric equilibria.}

In coordination games these equilibria are the only ones. Symmetric
discriminating equilibria only arise in anti-coordination games. The proof
of these results are based on the following lemma.

\begin{lemma}
\label{Lemmaalphai=alphaj}Consider a heterogeneous game $\Gamma
(n,T,(x_{t})_{t\in T},y,z)$. If there exists a symmetric equilibrium $%
(\alpha _{t}^{\ast })_{t\in T}$ with $0<\alpha _{i}^{\ast }<1$ and $0<\alpha
_{j}^{\ast }<1$ for $i,j\in T$ then $\alpha _{i}^{\ast }=\alpha _{j}^{\ast }$%
.
\end{lemma}

A direct consequence is that the probability of cooperation at equilibrium
can only take three different values: 0, 1 and an intermediate value.
Therefore there does not exist any $l$-discriminating equilibrium for $l>3$.

\begin{theorem}
\label{Theorem-NototallyDiscriminatingequilibrium}Consider a heterogeneous
game $\Gamma (n,T,(x_{t})_{t\in T},y,z)$. There is no totally discriminating
equilibrium for $m>3$.
\end{theorem}

Another less obvious consequence of the lemma is that there is no
discriminating equilibrium in coordination games.

\begin{theorem}
\label{PropCoordGames}Consider a heterogeneous game $\Gamma
(n,T,(x_{t})_{t\in T},y,z)$ with $y,z<0$. There is no discriminating
symmetric equilibrium.
\end{theorem}

The following theorems permit to find all discriminating equilibria in
anti-coordination games. For each partition $(T_{1},T_{2})$ of $T$, we
obtain one 2-discriminating equilibrium with $\alpha _{1}^{\ast }<\alpha
_{2}^{\ast }$.\footnote{%
The equilibrium with $\alpha _{1}^{\ast }>\alpha _{2}^{\ast }$ is obtained
for partition $(T_{2},T_{1})$ of $T$.}

\begin{theorem}
\label{Prop2-Discriminatingequilibrium}Consider a heterogeneous game $\Gamma
(n,T,(x_{t})_{t\in T},y,z)$ with $y,z>0$. Let ($T_{1},T_{2})$ be a partition
of $T$. The pair of strategies $((\alpha _{t}^{\ast })_{t\in T},(\alpha
_{t}^{\ast })_{t\in T})$ with 
\begin{equation*}
\alpha _{t}^{\ast }=\left\{ 
\begin{array}{cc}
\alpha _{1}^{\ast } & \text{if }t\in T_{1} \\ 
\alpha _{2}^{\ast } & \text{if }t\in T_{2}%
\end{array}%
\right. \text{ and }\alpha _{1}^{\ast }<\alpha _{2}^{\ast }\text{ }
\end{equation*}%
is an equilibrium of game $\Gamma (n,T,(x_{t})_{t\in T},y,z)$ if and only if
one of the following condition holds:

\begin{enumerate}
\item $\alpha _{1}^{\ast }=\frac{(n-1)\zeta -n\sum_{i\in T_{2}}x_{i}}{%
n\sum_{j\in T_{1}}x_{j}-1}$; $\alpha _{2}^{\ast }=1$; and $\sum_{i\in
T_{2}}x_{i}<\left( 1-\frac{1}{n}\right) \zeta $;

\item $\alpha _{1}^{\ast }=0$; $\alpha _{2}^{\ast }=1$; and $\left( 1-\frac{1%
}{n}\right) \zeta <\sum_{i\in T_{2}}x_{i}<\left( 1-\frac{1}{n}\right) \zeta +%
\frac{1}{n}$;

\item $\alpha _{1}^{\ast }=0$; $\alpha _{2}^{\ast }=\frac{(n-1)\zeta }{%
n\sum_{i\in T_{2}}x_{i}-1}$; and $\left( 1-\frac{1}{n}\right) \zeta +\frac{1%
}{n}<\sum_{i\in T_{2}}x_{i}.$
\end{enumerate}
\end{theorem}

Similarly, for each partition $(T_{1},T_{2},T_{3})$ of $T$, we obtain one
3-discriminating equilibrium with $\alpha _{1}^{\ast }<\alpha _{2}^{\ast
}<\alpha _{3}^{\ast }$.

\begin{theorem}
\label{Theorem3-Discriminatingequilibrium}Consider a heterogeneous game $%
\Gamma (n,T,(x_{t})_{t\in T},y,z)$ with $y,z>0$. Let $(T_{1},T_{2},T_{3})$
be a partition of $T$. The pair of strategies $((\alpha _{t}^{\ast })_{t\in
T},(\alpha _{t}^{\ast })_{t\in T})$ with 
\begin{equation*}
\alpha _{t}^{\ast }=\left\{ 
\begin{array}{cc}
0 & \text{if }t\in T_{1} \\ 
\frac{(n-1)\zeta -n\sum_{\substack{ k\in T_{3}}}x_{k}}{n\sum_{\substack{ %
j\in T_{2}}}x_{j}-1} & \text{if }t\in T_{2} \\ 
1 & \text{if }t\in T_{3}%
\end{array}%
\right.
\end{equation*}%
is an equilibrium of game $\Gamma (n,T,(x_{t})_{t\in T},y,z)$ if and only if 
\begin{equation*}
\sum_{k\in T_{3}}x_{k}<\left( 1-\frac{1}{n}\right) \zeta \ \text{and }%
\sum_{i\in T_{1}}x_{i}<\left( 1-\frac{1}{n}\right) \left( 1-\zeta \right) .
\end{equation*}
\end{theorem}

Note that the discriminating strategies at the equilibrium always include at
least one pure action. Another point worth to mention is that, when\textbf{\ 
}$T_{2}\neq \emptyset $ in Theorem \ref{Theorem3-Discriminatingequilibrium},
we obtain the results of Theorem \ref{Prop2-Discriminatingequilibrium} by
either setting $T_{1}=\emptyset $ or $T_{3}=\emptyset $. By contrast when $%
T_{2}=\emptyset $, Theorem \ref{Prop2-Discriminatingequilibrium} is not a
special case of Theorem\textbf{\ }\ref{Theorem3-Discriminatingequilibrium}%
\textbf{.}

\section{Discussion}

The main result of this paper is that when we have four types of players or
more in a heterogeneous game, there are no totally discriminating
equilibria: the maximum number of partitions where the treatment differs is
three. The Nash equilibria found in our results may be related to the
mechanism of human individual decision and the working memory, i.e., the few
temporarily active thoughts. The working memory is used in mental tasks,
problem solving and planning. In Cowan (2010), it has been discussed why the
number of items that an individual can store in the memory and remember for
a short period of time is around three, despite the reasons for that fact
remaining unclear in psychological science. The latter is in line with our
results on the ability to discriminate.

Such link between our theoretical results and those empirical ones discussed
in the latter provide an additional bridge between the fields of classic and
evolutionary game theory. Our results show that individuals are rationally
able to differentiate at most three partitions of types or three sets of
information. As Arthur (1994) points out, under complicated problems, the
type of rationality assumed in classic economics demands much of human
behavior and breaks down. Beyond a certain level of complexity, human
logical capacity ceases to cope and psychologists tend to agree that humans
think inductively with bounded rationality, simplifying the problem (Bower
and Hilgard, 1981; Holland et al., 1986; Rumelhart, 1980; Schank and
Abelson, 1977).

Thus, on the one hand, economic agents do rationally maximize their utility
or profit functions, on the other hand, the collection of information on the
possible ways that the utility function can be derived and built might be
too large for an individual to deal with, making him unable to identify all
the possibilities and ending up with a narrower set of strategies available
to choose from. As a consequence, when an agent chooses to play some
strategy, despite the fact that he selects the one that rationally maximizes
his utility, he is not fully aware if he is maximizing or not the utility
function that provides him with the largest possible maximized profit. This
creates room for the so-called bounded rationality in the literature. The
role of natural selection then links evolutionary and classic game theory in
dynamic models such as the replicator dynamics by selecting the
strategy(ies) which profit maximizing function(s) outperform(s) in the long
run, when the static stage-game is repeatedly played over time. Such
adaptive process replaces profit maximization at the individual level in
classic static games with profit maximization at the overall population
level in evolutionary dynamic games.

Certainly, other examples in different fields of science can be found and
related to our results, although in a further research paper we suggest the
focus might be on understanding how the results and the ways of
discriminating change when the players are aware of their own type as well.

\section*{Acknowledgments}

\noindent Laruelle acknowledges financial support from the Spanish Ministry
of Science and Innovation under funding PID2019-106146-I00 and from the
Basque Government (Research Group IT1367-19); Rocha acknowledges financial
support from the Brazilian Ministry of Science, Technology and Innovations
(CNPq funding 307437/2019-1).

\pagebreak

\section*{Appendix:}

\textbf{Equation (\ref{Ut2})}: Substituting (\ref{Probabilities}) in (\ref%
{Ut}), we obtain

\begin{eqnarray*}
U((\alpha _{t},\beta _{t})_{t\in T}) &=&\sum_{j\in T}\sum_{\substack{ i\in T 
\\ i\neq j}}\frac{nx_{i}x_{j}}{n-1}\left[ z\text{ }\beta _{i}+y\text{ }%
\alpha _{j}-(y+z)\beta _{i}\alpha _{j}\right] \\
&&+\sum_{j\in T}\frac{(nx_{j}-1)x_{j}}{n-1}\left[ z\text{ }\beta _{j}+y\text{
}\alpha _{j}-(y+z)\beta _{j}\alpha _{j}\right] \\
&=&\frac{1}{n-1}\sum_{j\in T}x_{j}\left[ \sum_{i\in T}nx_{i}-1\right] z\text{
}\beta _{j}+\frac{1}{n-1}\sum_{j\in T}x_{j}\left[ \sum_{i\in T}\text{ }%
nx_{i}-1\right] y\text{ }\alpha _{j} \\
&&-\frac{1}{n-1}\sum_{j\in T}n\text{ }x_{j}\text{ }(y+z)\left[ \sum_{j\in
T}x_{j}\beta _{j}\right] \alpha _{j}+\frac{1}{n-1}\sum_{j\in T}\text{ }%
x_{j}(y+z)\text{ }\beta _{j}\text{ }\alpha _{j} \\
&=&z\sum_{j\in T}x_{j}\text{ }\beta _{j}+\sum_{j\in T}\frac{x_{j}}{n-1}\left[
(n-1)y+(y+z)\beta _{j}-n(y+z)\sum_{i\in T}x_{i}\beta _{i}\right] \alpha _{j}.
\end{eqnarray*}%
\vspace{0.35cm}

\begin{proof}[\textbf{Proof of Theorem }\protect\ref%
{TheoremNonDiscriminatingequilibrium}]
Let $(\alpha ^{\ast },...,\alpha ^{\ast })$ be a non discriminating
equilibrium. By (\ref{FI(beta)}) we have $\mathcal{F}_{i}(\alpha ^{\ast
},...,\alpha ^{\ast })=(n-1)\left[ y-(y+z)\alpha ^{\ast }\right] $. We have $%
\alpha ^{\ast }=1$\ when $\mathcal{F}_{i}((\alpha _{t}^{\ast })_{t\in T})>0$
, i.e. $z<0$. We have $\alpha ^{\ast }=0$\ when $\mathcal{F}_{i}((\alpha
_{t}^{\ast })_{t\in T})<0$, i.e. $y<0$. If $0<\alpha ^{\ast }<1$ then $%
\mathcal{F}_{i}((\alpha _{t}^{\ast })_{t\in T})=0$ leads to $y-(y+z)\alpha
^{\ast }=0$, i.e. $\alpha ^{\ast }=\zeta $.
\end{proof}

\vspace{0.35cm}

\begin{proof}[Proof of Lemma \protect\ref{Lemmaalphai=alphaj}]
We must have $\mathcal{F}_{i}((\alpha _{t}^{\ast })_{t\in T})=0$ and $%
\mathcal{F}_{j}((\alpha _{t}^{\ast })_{t\in T})=0$. By (\ref{FI(beta)}) we
obtain $\mathcal{F}_{i}((\alpha _{t}^{\ast })_{t\in T})-\mathcal{F}
_{j}((\alpha _{t}^{\ast })_{t\in T})=(y+z)(\alpha _{i}^{\ast }-\alpha
_{j}^{\ast })=0.$ Given that $y+z\neq 0$, the only possibility is that $%
\alpha _{i}^{\ast }=\alpha _{j}^{\ast }$.
\end{proof}

\vspace{0.35cm}

\begin{proof}[Proof of \textbf{Theorem} \protect\ref%
{Theorem-NototallyDiscriminatingequilibrium}]
By Lemma \ref{Lemmaalphai=alphaj}, $\alpha _{i}^{\ast }=\alpha _{j}^{\ast
}=\alpha ^{\ast }$ with $0<\alpha ^{\ast }<1$ whenever $\mathcal{F}
_{i}((\alpha _{t}^{\ast })_{t\in T})=\mathcal{F}_{j}((\alpha _{t}^{\ast
})_{t\in T})=0$. The other values that $\alpha _{i}^{\ast }$ can take are $0$
(whenever $\mathcal{F}_{i}((\alpha _{t}^{\ast })_{t\in T})<0$) or $1$
(whenever $\mathcal{F}_{i}((\alpha _{t}^{\ast })_{t\in T})>0$). Given that $%
\alpha _{t}^{\ast }$ can only take three values: $0$, $1$ and $\alpha ^{\ast
}$, there is no totally discriminating equilibrium for $m>3$.
\end{proof}

\vspace{0.35cm}

\begin{proof}[Proof of \textbf{Theorem }\protect\ref{PropCoordGames}]
Suppose that there exists a discriminating symmetric equilibrium $(\alpha
_{t}^{\ast })_{t\in T}$ with $\alpha _{i}^{\ast }\neq \alpha _{j}^{\ast }$
for some $i,j\in T$. By Lemma \ref{Lemmaalphai=alphaj} $\alpha _{i}^{\ast
}\neq \alpha _{j}^{\ast }$ only holds in three cases: (1) $\alpha _{i}^{\ast
}=1$ and $\alpha _{j}^{\ast }=0$; (2) $\alpha _{i}^{\ast }=1$ and $0<\alpha
_{j}^{\ast }<1$; and (3) $\alpha _{i}^{\ast }=0$ and $0<\alpha _{j}^{\ast
}<1 $. (1) Suppose that $\alpha _{i}^{\ast }=1$ and $\alpha _{j}^{\ast }=0$.
Then we should have $\mathcal{F}_{i}((\alpha _{t}^{\ast })_{t\in
T})=(n-1)y+(y+z)-n(y+z)\sum_{t\in T}x_{t}\alpha _{t}^{\ast }>0$ and $%
\mathcal{F}_{j}((\alpha _{t}^{\ast })_{t\in T})=(n-1)y-n(y+z)\sum_{t\in
T}x_{t}\alpha _{t}^{\ast }<0$. This is impossible given that $y+z<0$. (2)
Suppose that $\alpha _{i}^{\ast }=1$ and $0<\alpha _{j}^{\ast }<1$. Then we
should have $\mathcal{F}_{i}((\alpha _{t}^{\ast })_{t\in
T})=(n-1)y+(y+z)-n(y+z)\sum_{t\in T}x_{t}\alpha _{t}^{\ast }>0$ and $%
\mathcal{F}_{j}((\alpha _{t}^{\ast })_{t\in T})=(n-1)y+(y+z)\alpha
_{j}^{\ast }-n(y+z)\sum_{t\in T}x_{t}\alpha _{t}^{\ast }=0$. This requires $%
(y+z)(1-\alpha _{j}^{\ast })>0$, which is impossible given that $y+z<0$ and $%
1-\alpha _{j}^{\ast }>0$. (3) Suppose that $\alpha _{i}^{\ast }=0$ and $%
0<\alpha _{j}^{\ast }<1$. Then we should have $\mathcal{F}_{i}((\alpha
_{t}^{\ast })_{t\in T})=(n-1)y-n(y+z)\sum_{t\in T}x_{t}\alpha _{t}^{\ast }<0$
and $\mathcal{F}_{j}((\alpha _{t}^{\ast })_{t\in T})=(n-1)y+(y+z)\alpha
_{j}^{\ast }-n(y+z)\sum_{t\in T}x_{t}\alpha _{t}^{\ast }=0$. This requires $%
(y+z)\alpha _{j}^{\ast }>0$, which is again impossible given that $y+z<0$
and $\alpha _{j}^{\ast }>0$. In order to have a discriminating symmetric
equilibrium, a necessary condition is that $y+z>0$.
\end{proof}

\vspace{0.35cm}

\begin{proof}[Proof of \textbf{Theorem }\protect\ref%
{Prop2-Discriminatingequilibrium}]
Suppose that there exists a 2-discriminating symmetric equilibrium $(\alpha
_{t}^{\ast })_{t\in T}$ with $\alpha _{t}^{\ast }=\alpha _{1}^{\ast }$ for $%
t\in T_{1}$ and $\alpha _{t}^{\ast }=\alpha _{2}^{\ast }$ for $t\in T_{2}$
and $\alpha _{1}^{\ast }<\alpha _{2}^{\ast }$. There are only three possible
cases: (1) $0<\alpha _{1}^{\ast }<1$ and $\alpha _{2}^{\ast }=1$; (2) $%
\alpha _{1}^{\ast }=0$ and $\alpha _{2}^{\ast }=1$; (3) $\alpha _{1}^{\ast
}=0$ and $0<\alpha _{2}^{\ast }<1$. (1) We have $0<\alpha _{1}^{\ast }<1$
and $\alpha _{2}^{\ast }=1$ if $\mathcal{F}_{j}((\alpha _{t}^{\ast })_{t\in
T})=0$ for $j\in T_{1}$, which gives by (\ref{FI(beta)}) $\alpha _{1}^{\ast
}=\frac{(n-1)\zeta -n\sum_{i\in T_{2}}x_{i}}{n\sum_{j\in T_{1}}x_{j}-1},$
and $\mathcal{F}_{i}((\alpha _{t}^{\ast })_{t\in T})>0$ for $i\in T_{2}$.
The condition $\alpha _{1}^{\ast }>0$ gives $\sum_{i\in T_{2}}x_{i}<\left( 1-%
\frac{1}{n}\right) \zeta $. We can check that $\alpha _{1}^{\ast }<1$ and $%
\mathcal{F}_{i}((\alpha _{t}^{\ast })_{t\in T})>0$ for $i\in T_{2}$. (2)
Similarly we have $\alpha _{1}^{\ast }=0$ if $\mathcal{F}_{j}((\alpha
_{t}^{\ast })_{t\in T})<0$ for $j\in T_{1}$, which requires $\sum_{i\in
T_{2}}x_{i}>\left( 1-\frac{1}{n}\right) \zeta $. Similarly we have $\alpha
_{2}^{\ast }=1$ if $\mathcal{F}_{i}((\alpha _{t}^{\ast })_{t\in T})>0$ for $%
i\in T_{2},$ which gives $\sum_{i\in T_{2}}x_{i}<\left( 1-\frac{1}{n}\right)
\zeta +\frac{1}{n}.$ Plugging both conditions together, $\left( 1-\frac{1}{n}%
\right) \zeta <\sum_{i\in T_{2}}x_{i}<\left( 1-\frac{1}{n}\right) \zeta +%
\frac{1}{n}$. (3) We have $\alpha _{1}^{\ast }=0$ and $0<\alpha _{2}^{\ast
}<1$ if $\mathcal{F}_{i}((\alpha _{t}^{\ast })_{t\in T})=0$ for $i\in T_{2}$
and $\mathcal{F}_{j}((\alpha _{t}^{\ast })_{t\in T})<0$ for $j\in T_{1}$.
The equality $\mathcal{F}_{i}((\alpha _{t}^{\ast })_{t\in T})=0$ leads to $%
\alpha _{2}^{\ast }=\frac{(n-1)\zeta }{n\sum_{i\in T_{2}}x_{i}-1}$. Clearly $%
\alpha _{2}^{\ast }>0\ $ and the inequality $\alpha _{2}^{\ast }<1$ requires 
$\left( 1-\frac{1}{n}\right) \zeta +\frac{1}{n}<\sum_{i\in T_{2}}x_{i}$. We
can easily check that $\mathcal{F}_{j}((\alpha _{t}^{\ast })_{t\in T})<0$
for $j\in T_{1}$.
\end{proof}

\vspace{0.35cm}

\begin{proof}[Proof of \textbf{Theorem }\protect\ref%
{Theorem3-Discriminatingequilibrium}]
Suppose that there exists a discriminating symmetric equilibrium $((\alpha
_{t}^{\ast })_{t\in T},(\alpha _{t}^{\ast })_{t\in T})$ with $\alpha
_{t}^{\ast }=\alpha _{1}^{\ast }$ for $t\in T_{1}$ and $\alpha _{t}^{\ast
}=\alpha _{2}^{\ast }$ for $t\in T_{2};$ and $\alpha _{t}^{\ast }=\alpha
_{3}^{\ast }$ for $t\in T_{3}$; $\alpha _{1}^{\ast }<\alpha _{2}^{\ast
}<\alpha _{3}^{\ast }
$. Then there is only one possibility: $\alpha _{1}^{\ast }=0\ $and $\alpha
_{2}^{\ast }=\lambda ^{\ast }$ and $\alpha _{3}^{\ast }=1$. The following
conditions must hold: (i) $\mathcal{F}_{i}((\alpha _{t}^{\ast })_{t\in T})<0$
for $i\in T_{1}$; (ii) $\mathcal{F}_{j}((\alpha _{t}^{\ast })_{t\in T})=0$
for $j\in T_{2}$; and (iii) $\mathcal{F}_{k}((\alpha _{t}^{\ast })_{t\in
T})>0$ for $k\in T_{3}$. Condition $\mathcal{F}_{j}((\alpha _{t}^{\ast
})_{t\in T})=0$ for $j\in T_{2}$ gives $(n-1)y+(y+z)\lambda ^{\ast
}-n(y+z)\sum_{t\in T}x_{t}\alpha _{t}^{\ast }=0$. After some algebra, we
obtain $\lambda ^{\ast }=\frac{(n-1)\zeta -n\sum_{k\in T_{3}}x_{k}}{%
n\sum_{j\in T_{2}}x_{j}-1}$. The condition $\lambda ^{\ast }<1$ can be
written as $(n-1)\zeta -n\sum_{k\in T_{3}}x_{k}<n\sum_{j\in T_{2}}x_{j}-1.$
This gives the condition $\sum_{i\in T_{1}}x_{i}<\left( 1-\frac{1}{n}\right)
(1-\zeta )$. Condition $\lambda ^{\ast }>0$ can be written as $(n-1)\zeta
-n\sum_{k\in T_{3}}x_{k}>0$ or $\sum_{k\in T_{3}}x_{k}<\left( 1-\frac{1}{n}%
\right) \zeta $. Finally two conditions remain to be checked: $\mathcal{F}%
_{i}((\alpha _{t}^{\ast })_{t\in T})<0$ for $i\in T_{1}$ and $\mathcal{F}%
_{k}((\alpha _{t}^{\ast })_{t\in T})>0$ for $k\in T_{3}$. Note that for $%
i\in T_{1}$, $j\in T_{2}$ and $k\in T_{3}$ we have $\mathcal{F}_{i}((\alpha
_{t}^{\ast })_{t\in T})-\mathcal{F}_{j}((\alpha _{t}^{\ast })_{t\in
T})=-(y+z)\lambda ^{\ast }<0$ and $\mathcal{F}_{k}((\alpha _{t}^{\ast
})_{t\in T})-\mathcal{F}_{j}((\alpha _{t}^{\ast })_{t\in T})=(1-\lambda
^{\ast })(y+z)>0.$ Given that $\mathcal{F}_{j}((\alpha _{t}^{\ast })_{t\in
T})=0$, we have $\mathcal{F}_{i}((\alpha _{t}^{\ast })_{t\in T})<0$ for $%
i\in T_{1}$ and $\mathcal{F}_{k}((\alpha _{t}^{\ast })_{t\in T})>0$ for $%
k\in T_{3}$.
\end{proof}

\end{document}